\documentclass[11pt]{article}

\usepackage{graphicx}
\usepackage{amssymb}
\usepackage{amsmath}
\usepackage{amsthm}
\usepackage{array}
\newcolumntype{L}[1]{>{\raggedright\let\newline\\\arraybackslash\hspace{0pt}}m{#1}}
\newcolumntype{C}[1]{>{\centering\let\newline\\\arraybackslash\hspace{0pt}}m{#1}}
\newcolumntype{R}[1]{>{\raggedleft\let\newline\\\arraybackslash\hspace{0pt}}m{#1}}

\usepackage{gastex}
\usepackage{eucal}
\usepackage{url}

\usepackage{cite}
\usepackage{float}
\usepackage{subfig}

\usepackage{multirow}
\usepackage{rotating}
\usepackage{tikz}
\usetikzlibrary{arrows,automata, positioning}
\definecolor{light-gray}{gray}{0.8}

\newtheorem{Thm}{Theorem}
\newtheorem{Prop}[Thm]{Proposition}
\newtheorem{Lemma}[Thm]{Lemma}
\newtheorem{Cor}[Thm]{Corollary}

\DeclareSymbolFont{rsfscript}{OMS}{rsfs}{m}{n}
\DeclareSymbolFontAlphabet{\mathrsfs}{rsfscript}

\newcommand{\dt}{.}
\DeclareMathOperator{\ret}{\mathrm{rt}}

\newcommand{\sa}{synchronizing automata}
\newcommand{\sna}{synchronizing $n$-automata}
\newcommand{\snan}{synchronizing $n$-au\-tom\-a\-ton}
\newcommand{\san}{synchronizing automaton}
\newcommand{\scn}{strongly connected}
\newcommand{\sca}{strongly connected automata}
\newcommand{\scan}{strongly connected automaton}
\newcommand{\sw}{reset word}
\newcommand{\sws}{reset words}

\newcommand{\rt}{reset threshold}

\newcommand{\mA}{\mathrsfs{A}}
\newcommand{\mB}{\mathrsfs{B}}

\newcommand{\mF}{\mathrsfs{F}}
\newcommand{\mI}{\mathrsfs{I}}
\newcommand{\mS}{\mathrsfs{S}}

\begin{document}%

\title{Slowly Synchronizing Automata\\ with Idempotent Letters of Low Rank}

\author{Mikhail V. Volkov\thanks{Supported by the Ministry of Science and Higher Education of the Russian Federation, project no.\ 1.580.2016, and the Competitiveness Enhancement Program of Ural Federal University.}}

\date{Institute of Natural Sciences and Mathematics\\ Ural Federal University\\
Ekaterinburg, Russia\\
m.v.volkov@urfu.ru}

\maketitle

\begin{abstract}
   We use a semigroup-theoretic construction by Peter Higgins in order to produce, for each even $n$, an $n$-state and 3-letter \san\ with the following two features:\\
1) all its input letters act as idempotent selfmaps of rank $\frac{n}2$;\\
2) its \rt\ is asymptotically equal to $\frac{n^2}2$.
\end{abstract}

\section{Background and overview}
\label{sec:intro}

A \emph{complete deterministic finite automaton} (DFA) is a triple $\langle Q,\Sigma,\delta\rangle$, where $Q$ and $\Sigma$ are finite sets called the \emph{state set} and the \emph{input alphabet} respectively, and  $\delta\colon Q\times\Sigma\to Q$ is a totally defined map called the \emph{transition function}. Let $\Sigma^*$ stand for the collection of all finite words over the alphabet $\Sigma$, including the empty word. The transition function extends to a function $Q\times\Sigma^*\to Q$, still denoted $\delta$, in the following natural way: for every $q\in Q$ and $w\in\Sigma^*$, we set $\delta(q,w):=q$ if $w$ is empty and $\delta(q,w):=\delta(\delta(q,v),a)$ if $w=va$ for some $v\in\Sigma^*$ and some $a\in\Sigma$. Thus, every word $w\in\Sigma^*$ induces the selfmap $q\mapsto\delta(q,w)$ of the set $Q$; we say that $w$ is \emph{idempotent} if so is the selfmap induced by $w$, that is, if $\delta(q,w)=\delta(q,w^2)$ for each $q\in Q$.

When we deal with a fixed DFA, we simplify our notation by suppressing the sign of the transition function; this means that we introduce the DFA as a pair $\langle Q,\Sigma\rangle$ rather than a triple $\langle Q,\Sigma,\delta\rangle$ and write $q\dt w$ for $\delta(q,w)$ and $Q\dt w$ for $\{\delta(q,w)\mid q\in Q\}$.

A DFA $\mA=\langle Q,\Sigma\rangle$ is called \emph{synchronizing} if there exists a word $w\in\Sigma^*$ whose action \emph{resets} $\mA$, that is, $w$ leaves the automaton
in one fixed state, regardless of the state at which $w$ is applied. This means that $q\dt w=q'\dt w$ for all $q,q'\in Q$. Any word $w$ with this property is said to be a \emph{reset} word for the automaton, and the minimum length of \sws\ for $\mA$, denoted $\ret(\mA)$, is called the \emph{reset threshold} of $\mA$.

Synchronizing automata serve as transparent and useful models of error-resistant systems in many applied areas (system and protocol testing, information coding, robotics). At the same time, \sa\ surprisingly arise in some parts of pure mathematics (symbolic dynamics, theory of substitution systems, and others). Basics of the theory of \sa\ as well as its diverse connections and applications are discussed, for instance, in the survey~\cite{Volkov:2008} and in the chapter~\cite{KV} of the forthcoming ``Handbook of Automata Theory''. Here we focus on only one aspect of the theory, namely, on the question of how the \rt\ of a \san\ depends on the number of states.

A DFA with two input letters is called \emph{binary}. In~1964, \v{C}ern\'{y}~\cite{Cerny:1964} constructed for each $n>1$, a binary \san\ $\mathrsfs{C}_n$ with $n$ states and \rt\ $(n-1)^2$. Recall the definition of $\mathrsfs{C}_n$. If we denote the states of $\mathrsfs{C}_n$ by $1,2,\dots,n$ and the input letters by $\sigma_1$ and $\sigma_2$, the actions of the letters are as follows:
\[
i\dt\sigma_1:=\begin{cases}
i &\text{if } i<n,\\
1 &\text{if } i=n;
\end{cases}
\qquad
i\dt\sigma_2:=\begin{cases}
i+1 &\text{if } i<n,\\
1 &\text{if } i=n.
\end{cases}
\]
The automaton $\mathrsfs{C}_n$ is shown in Fig.\,\ref{fig:cerny-n}.
\begin{figure}[ht]
\begin{center}
\unitlength .55mm
\begin{picture}(72,90)(0,-83)
\gasset{Nw=16,Nh=16,Nmr=8}
\node(n0)(36.0,-16.0){1}
\node(n1)(4.0,-40.0){$n$} \node(n2)(68.0,-40.0){2}
\node(n3)(16.0,-72.0){$n{-}1$} \node(n4)(56.0,-72.0){3}
\drawedge[ELdist=2.0](n1,n0){$\sigma_1,\sigma_2$} \drawedge[ELdist=1.5](n2,n4){$\sigma_2$}
\drawedge[ELdist=1.7](n0,n2){$\sigma_2$}
\drawedge[ELdist=1.7](n3,n1){$\sigma_2$}
\drawloop[ELdist=1.5,loopangle=30](n2){$\sigma_1$}
\drawloop[ELdist=2.4,loopangle=-30](n4){$\sigma_1$}
\drawloop[ELdist=1.5](n0){$\sigma_1$}
\drawloop[ELdist=1.5,loopangle=210](n3){$\sigma_1$}
\put(31,-73){$\dots$}
\end{picture}
\caption{The automaton $\mathrsfs{C}_n$}\label{fig:cerny-n}
\end{center}
\end{figure}

The automata in the \v{C}ern\'{y} series are well-known in the connection with the famous \v{C}ern\'{y} conjecture about the maximum \rt\ for \sa\ with $n$ states, see \cite{Volkov:2008}. The automata $\mathrsfs{C}_n$ provide the lower bound $(n-1)^2$ for this maximum, and the conjecture claims that these automata represent the worst possible case since it has been conjectured that every \san\ with $n$ states can be reset by a word of length $(n-1)^2$. The conjecture, first stated in the 1960s, resists researchers' efforts for more than 50 years. The best upper bound achieved so far is cubic in $n$; it is due to Shitov~\cite{Shitov:2019} who has slightly improved the bound established by  Szyku\l{}a~\cite{Szykula:2018}. In turn, Szyku\l{}a's bound is only slightly better than the upper bound $\frac{n^3-n}6$ established by Pin~\cite{Pin:1983} and Frankl~\cite{Frankl:1982} approx.\ 35 years ago.

Why is the \v{C}ern\'{y} conjecture so surprisingly hard? Here we mention only one of the difficulties encountered by the theory of \sa, namely, the shortage of examples of
\emph{slowly \sa}, i.e., automata with \rt\ close to the square of the number of states. It has already been observed in the literature that with a very restricted number of examples in hand, it was hard to verify various guesses and assumptions that had arisen when researchers were searching for approaches to the \v{C}ern\'{y} conjecture. That is why the history of investigations in the area abounds in  ``false trails'', i.e., auxiliary hypotheses that looked promising at first but were disproved after some time.

For brevity, a DFA with $n$ states is referred to as an $n$-\emph{automaton}. The series found in~\cite{Cerny:1964} still remains the only known infinite series of $n$-automata with \rt\ $(n-1)^2$. Besides that, we know only a few isolated examples of such automata, the largest (with respect to the state number) being the 6-automaton discovered by Kari~\cite{Kari:2001}; see~\cite{Volkov:2008} for a complete list of known \sna\ with \rt\ $(n-1)^2$. Moreover, even infinite series of \sna\ whose \rt{}s are \emph{asymptotically} equal to $(n-1)^2$ turns out to be extremely rare, especially those consisting of non-binary automata. We call a non-binary \san\ $\mA=\langle Q,\Sigma,\delta\rangle$ \emph{proper} if for every letter $a\in\Sigma$, the DFA $\langle Q,\Sigma^{a},\delta^{a}\rangle$ is not synchronizing, where $\Sigma^{a}:=\Sigma\setminus\{a\}$ and $\delta^{a}$ is the restriction of $\delta$ to the set $Q\times\Sigma^{a}$; in other terms, $\mA$ is proper whenever $|\Sigma|>2$ and each \sw\ of $\mA$ involves all letters of $\Sigma$. The largest known proper DFA with \rt\ that matches the \v{C}ern\'{y}'s bound is the 3-letter 5-automaton found by Roman~\cite{Roman:2008}. We mention also two infinite series of proper 3-letter \sna\ with \rt{}s $n^2-3n+3$ and $n^2-3n+2$ constructed by Kisielewicz and Szyku\l{}a~\cite{Kisielewicz&Szykula:2015} and two infinite series of proper \sna\ with $k$ letters ($k$ is any integer greater than 2) and  \rt{}s $n^2-(k+2)n+2k+1$ and $n^2-(k+2)n+2k+2$ invented by D\.zyga, Ferens, Gusev, and Szyku\l{}a~\cite[Section 6]{Dzyga:2017}.

It appears that the border value of \rt, at which one begins to observe proper non-binary \sna\ more frequently, is situated somewhere near the value $\frac{n^2}2$. In the literature, one can find several constructions of infinite series of non-binary \sna\ such that their \rt{}s are asymptotically equal to $\frac{n^2}2$ and, besides that, the automata in the series possess some specific property such as having a sink state~\cite{Rystsov:1997}, admitting a directed Eulerian walk~\cite{Szykula&Vorel:2016}, or being capable to perform an arbitrary selfmap of the state set~\cite{GGGJV:2019}. For further examples of proper non-binary \sa\ with relatively large \rt, we refer to~\cite{Catalano:2018,Pribavkina:2011}.

In the present note we provide a tool for constructing new series of \sa\ with similar parameters but rather peculiar extra properties. Namely, we describe a simple transformation that, given an arbitrary \snan\ $\mA$ with $k$ input letters, produces a \san\ $\mathbb{H}(\mA)$ with $2n$ states and $k+1$ input letters such that every letter of $\mathbb{H}(\mA)$ acts as an idempotent selfmap on the state set and $\ret(\mathbb{H}(\mA))=2\ret(\mA)$. If applied to a \snan\ whose \rt\ is close to $(n-1)^2$, the transformation results in a \san\ with $m:=2n$ states and \rt\ close to $\frac{m^2}2$. In particular, if one applies it to the automata $\mathrsfs{C}_n$ in the \v{C}ern\'{y} series, one gets a new series of proper \sa\ $\mB_m:=\mathbb{H}(\mathrsfs{C}_n)$ with $m:=2n$ states and 3 idempotent input letters such that the \rt\ of $\mB_m$ is equal to $2(n-1)^2=\frac{m^2}2-2m+2$.

An additional feature of the transformation $\mA\mapsto\mathbb{H}(\mA)$ is that it leads to automata in which all letters have relatively low rank. Recall that the \emph{rank} of a letter $a\in\Sigma$ with respect to a DFA $\mA=\langle Q,\Sigma\rangle$ is defined as the cardinality of the set $Q\dt a$. In an overwhelming majority of known examples of slowly \sna, their input letters all have ranks close to $n$. Gusev~\cite{Gusev:2012} investigated the following question: if all input letters of a \sa\ $\mA=\langle Q,\Sigma\rangle$ have rank at most $r<|Q|$, how does the \rt\ of $\mA$ depend on $r$? He provided a lower bound (about which he expressed the hope that it might be close to optimal) via rather a tricky construction that yields for each $r$, a binary \san\  whose letters have rank $r$ and whose reset threshold is at least $r^2-r-1$. In our automata of the form $\mathbb{H}(\mA)$, the rank $r$ of each input letter is equal to one half of number of states, whence, as we see from the aforementioned example of the series $\mB_m:=\mathbb{H}(\mathrsfs{C}_n)$, the \rt\ of such automata as a function of $r$ may attain the value $2r^2-4r+2$. This improves the lower bound from~\cite{Gusev:2012}; we recall, however, that automata in Gusev's construction are binary while the automata $\mB_m$ are not.

Our transformation $\mA\mapsto\mathbb{H}(\mA)$ is a straightforward adaptation of a construction suggested by Higgins~\cite{Higgins} in the realm of semigroup theory. Higgins used this construction to give a new simple proof of the following result independently obtained in \cite{GH:1984} and \cite{Laffey:1983}: an arbitrary (finite) semigroup may be embedded into another (finite) semigroup in which every element is the product of two idempotents. Our contribution consists in observing that, when restated in automata-theoretic terms, Higgins's construction leads to a transformation that preserves both the property of being synchronizing and the order of magnitude of \rt.

We describe the transformation and prove its main properties in Section~2. In Section~3 we establish a tight upper bound for the \rt{} of \sa\ with two idempotent input letters. Here the \rt{} is always less than the number of states, in a strong contrast to the non-binary case.

A preliminary version of this paper appeared in July 2018 as preprint~\cite{Volkov:2018} and was submitted to a journal in October 2018. Later the author learned that the transformation presented in Section~2 was independently found by Don and Zantema; their preprint~\cite{Don&Zantema:2018} appeared in December 2018.

\section{The transformation}

Take a DFA $\mathrsfs{A}=\langle Q,\Sigma\rangle$ with $Q:=\{1,2,\dots,n\}$ and $\Sigma:=\{\sigma_1,\dots,\sigma_k\}$. Let $Q':=\{1',2',\dots,n'\}$ be a copy of $Q$ such that $Q$ and $Q'$ is disjoint. Consider the automaton $\mathbb{H}(\mathrsfs{A})$ with the state set $R:=Q\cup Q'$ and the input alphabet $\Theta:=\{a_1,\dots,a_k,b\}$ in which the action of the letters is defined as follows:
\begin{gather}
\label{eq:a} i\dt a_j:=i\ \text{ and }\ i'\dt a_j:=i\dt\sigma_j\ \text{ for all }\ i=1,\dots,n,\ j=1,\dots,k;\\
\label{eq:b} i\dt b=i'\dt b:=i'\ \text{ for all }\ i=1,\dots,n.
\end{gather}
For an illustration, Fig.~\ref{fig:duplication} shows the automaton that one gets if the above construction is applied to the \v{C}ern\'{y} automaton $\mathrsfs{C}_n$ from Fig.~\ref{fig:cerny-n}.

\begin{figure}[th]
\begin{center}
    \unitlength=.85mm
       \begin{picture}(140,90)(0,0)
        \rpnode(qn2)(10,10)(20,7){$(n{-}1)'$}
        \rpnode(qn1)(28,52)(20,7){$n'$}
        \rpnode(q0)(70,70)(20,7){$1'$}
        \rpnode(q1)(112,52)(20,7){$2'$}
        \rpnode(q2)(130,10)(20,7){$3'$}
        \rpnode(qn2_)(35,10)(20,7){$n{-}1$}
        \rpnode(qn1_)(45,35)(20,7){$n$}
        \rpnode(q0_)(70,45)(20,7){$1$}
        \rpnode(q1_)(95,35)(20,7){$2$}
        \rpnode(q2_)(105,10)(20,7){$3$}
        \drawedge[curvedepth=5](qn2,qn1_){$a_2$}
        \drawedge[curvedepth=5](q0,q1_){$a_2$}
        \drawedge[curvedepth=5](q1,q2_){$a_2$}
        \drawedge[curvedepth=5](qn2,qn2_){$a_1$}
        \drawedge[curvedepth=5](qn1,q0_){$a_1,a_2$}
        \drawedge[curvedepth=5](q0,q0_){$a_1$}
        \drawedge[curvedepth=5](q1,q1_){$a_1$}
        \drawedge[curvedepth=5](q2,q2_){$a_1$}
        \drawedge[curvedepth=5](qn2_,qn2){$b$}
        \drawedge[curvedepth=5](qn1_,qn1){$b$}
        \drawedge[curvedepth=5](q0_,q0){$b$}
        \drawedge[curvedepth=5](q1_,q1){$b$}
        \drawedge[curvedepth=5](q2_,q2){$b$}
        \drawloop[loopangle=270](q0_){$a_1,a_2$}
        \drawloop[loopangle=225](q1_){$a_1,a_2$}
        \drawloop[loopangle=180](q2_){$a_1,a_2$}
        \drawloop[loopangle=-45](qn1_){$a_1,a_2$}
        \drawloop[loopangle=0](qn2_){$a_1,a_2$}
        \drawloop[loopangle=90](q0){$b$}
        \drawloop[loopangle=45](q1){$b$}
        \drawloop[loopangle=0](q2){$b$}
        \drawloop[loopangle=135](qn1){$b$}
        \drawloop[loopangle=180](qn2){$b$}
        \put(7,0){$\dots$}
        \put(128,0){$\dots$}
        \put(32,0){$\dots$}
        \put(103,0){$\dots$}
    \end{picture}
    \caption{The automaton $\mB_m:=\mathbb{H}(\mathrsfs{C}_n)$}\label{fig:duplication}
   \end{center}
\end{figure}

We start with registering two properties of the automaton $\mathbb{H}(\mathrsfs{A})$ that immediately follow from the definitions~\eqref{eq:a} and~\eqref{eq:b}.

\begin{Lemma}
\label{lem:idempotent}
For every DFA $\mathrsfs{A}$, the DFA\/ $\mathbb{H}(\mathrsfs{A})=\langle R,\Theta\rangle$ is such that each letter in $\Theta$ acts on $R$ as an idempotent selfmap and has rank  $\frac{|R|}2$.
\end{Lemma}

Now we establish the relation between \rt{}s of $\mA$ and $\mathbb{H}(\mathrsfs{A})$.

\begin{Thm}
\label{thm:main}
The automaton\/ $\mathbb{H}(\mathrsfs{A})$ is synchronizing if and only if so is $\mathrsfs{A}$,  and if this is the case, then $\ret(\mathbb{H}(\mA))=2\ret(\mA)$.
\end{Thm}

\begin{proof}
First suppose that $\mathrsfs{A}=\langle Q,\Sigma\rangle$ with $Q=\{1,2,\dots,n\}$ and $\Sigma=\{\sigma_1,\dots,\sigma_k\}$ is a \san. Let $w\in\Sigma^*$ be a \sw\ of $\mA$ of minimum length. Consider the morphism $\chi\colon\Sigma^*\to\Theta^*$ defined by the rule $\chi(\sigma_j):=ba_j$ for $j=1,\dots,k$ and let $W:=\chi(w)$. Clearly, $|W|=2|w|$. From~\eqref{eq:a} and~\eqref{eq:b} we readily conclude that
\begin{equation}
\label{eq:imitation}
i.\sigma_j=i.\chi(\sigma_j)=i'.\chi(\sigma_j)
\end{equation}
for all $i=1,\dots,n$ and all $j=1,\dots,k$. From~\eqref{eq:imitation}, we get $Q\dt w=R\dt W$, whence $W$ is a \sw\ for the automaton $\mathbb{H}(\mathrsfs{A})$. We see that $\mathbb{H}(\mathrsfs{A})$ is synchronizing and
\begin{equation}
\label{eq:doubling}
\ret(\mathbb{H}(\mA))\le2\ret(\mA)
\end{equation}
since $\ret(\mathbb{H}(\mA))\le|W|=2|w|$ and $|w|=\ret{\mA}$ by the choice of the word $w$.

Conversely, suppose that the automaton $\mathbb{H}(\mathrsfs{A})$ is synchronizing, and let $U\in\Theta^*$ be its \sw\ of minimum length so that $\ret(\mathbb{H}(\mA))=|U|$. Since the letter $b$ acts as an idempotent on $R$, the word $b^2$ does not occur in $U$ as a factor. (Otherwise we could reduce $b^2$ to $b$ and obtain a shorter \sw.) Further, since by~\eqref{eq:a} the image of each of the letters $a_1,\dots,a_k$ is contained in $Q$ and the letters $a_1,\dots,a_k$ act as the identity map on $Q$, we conclude that $r\dt a_sa_t=r\dt a_s$ for all $r\in R$ and all $s,t\in\{1,\dots,k\}$. Therefore, none of the words $a_sa_t$ can occur in $U$ as a factor. Thus, in $U$, every occurrence of the letter $b$, except its possible occurrence as the last letter of $U$, is followed by an occurrence of one of the letters $a_1,\dots,a_k$, and for each $j=1,\dots,k$, every occurrence of the letter $a_j$, except its possible occurrence as the last letter of $U$, is followed by an occurrence of the letter $b$. In other terms, the occurrences of $b$ in the word $U$ alternate with the occurrences of $a_1,\dots,a_k$.

Suppose that $b$ is the last letter of $U$. Then $U=U'b$ for some $U'\in \Theta^*$ and $U'$ does not reset $\mathbb{H}(\mathrsfs{A})$ since $|U'|<|U|$ and $U$ was chosen to be a \sw\ of minimum length. Therefore, $|R\dt U'|\ge2$. The last letter of $U'$ is one of the letters $a_1,\dots,a_k$, whence $R\dt U'\subseteq Q$. However, since the letter $b$ bijectively maps $Q$ onto $Q'$, it cannot merge the states in any subset of $Q$. Hence, $|R\dt U|=|R\dt U'b|=|R\dt U'|\ge2$, and this contradicts our choice of the word $U$. Thus, $U$ cannot end with $b$.

Now we aim to show that $U$ starts with $b$. Again, we argue by contradiction. Suppose that $U=a_jU''$ for some $j\in\{1,\dots,k\}$ and some $U''\in \Theta^*$. Observe that $U''$ starts with $b$, does not end with $b$, and the occurrences of $b$ in $U''$ alternate with the occurrences of $a_1,\dots,a_k$. Thus, $U''$ can be decomposed as a product of factors of the form $ba_j$, $j=1,\dots,k$, that is, $U''$ belongs to the image of the morphism $\chi\colon\Sigma^*\to\Theta^*$ introduced in the first paragraph of the proof. Consider the word $v\in\Sigma^*$ defined by $v:=\chi^{-1}(U'')$. It is clear that $|v|=\frac{|U''|}2$, and from \eqref{eq:imitation}, it readily follows that $Q\dt v=Q\dt U''$. By~\eqref{eq:a}, the letter $a_j$ fixes all states in $Q$, whence $Q\dt U''=Q\dt a_jU''=Q\dt U\subseteq R\dt U$. Since $U$ is a \sw\ for $\mathbb{H}(\mathrsfs{A})$, the set $R\dt U$ is a singleton and so is the set $Q\dt v=Q\dt U''$. We conclude that $v$ is a \sw\ for $\mA$, whence $\ret(\mA)\le|v|$. Using the inequality \eqref{eq:doubling}, we arrive at the conclusion that
\[
\ret(\mathbb{H}(\mA))\le 2\ret(\mA)\le 2|v|=|U''|<|U|,
\]
which contradicts the choice of the word $U$.

Summing up the facts established so far, the word $U$ starts with the letter $b$, ends with one of the letters $a_1,\dots,a_k$, and  the occurrences of $b$ in $U$ alternate with the occurrences of $a_1,\dots,a_k$. This ensures that $U$ is a product of factors of the form $ba_j$, $j=1,\dots,k$, and we can apply to $U$ the inverse of the morphism $\chi$ as we did in the preceding paragraph for the word $U''$. Let $u:=\chi^{-1}(U)$. Then $|u|=\frac{|U|}2$, and from \eqref{eq:imitation} we conclude that $Q\dt u=R\dt U$. Hence, $u$ is a \sw\ for $\mA$, and $\ret(\mA)\le |u|=\frac{|U|}2=\frac{\ret(\mathbb{H}(\mA))}2$. Combining this inequality with \eqref{eq:doubling}, we obtain the equality $\ret(\mathbb{H}(\mA))=2\ret(\mA)$.
\end{proof}

By Lemma~\ref{lem:idempotent} and Theorem~\ref{thm:main}, the transformation $\mA\mapsto\mathbb{H}(\mA)$  applied to the automata in \v{C}ern\'{y}'s series produces a series of automata that witnesses the following fact:

\begin{Cor}
\label{cor:cerny}
For each even $n$, there exists a $3$-letter proper \snan\ with \rt\ $\frac{n^2}2-2n+2$ whose letters are idempotents of rank $\frac{n}2$.
\end{Cor}

\section{The binary case}
A DFA with one input letter is called \emph{unary}. It is known and easy to very that for every unary \san, its \rt\ is strictly less than the number of states. Therefore the application of the transformation $\mA\mapsto\mathbb{H}(\mA)$ to a unary \san\ cannot produce a binary \san\ with idempotent input letters and \rt\ greater than or equal to the number of states. This does not exclude the possibility that such binary \sa\ can be constructed in some other way. In this section, we address this issue. We prove that the \rt{} of every \snan\ with two idempotent input letters does not exceed $n-1$ and this bound is tight. The upper bound $n-1$ has been established for several types of \sna, see, e.g., \cite{Ananichev:2005,Rystsov:1997}, but it appears that automata with two idempotent input letters do not fall into any previously analyzed class.

The following arguments involve a few standard concepts of automata theory which we recall here for the sake of completeness.

Let $\mathrsfs{A}=\langle Q,\Sigma,\delta\rangle$ be a DFA. If $S\subseteq Q$ is such that $\delta(s,a)\in S$ for all $s\in S$ and $a\in\Sigma$, one can consider the DFA $\mathrsfs{S}:=\langle S,\Sigma,\tau\rangle$, where the function $\tau$ is defined as the restriction of $\delta$ to the set $S\times\Sigma$, that is, $\tau(s,a):=\delta(s,a)$ for all $s\in S$ and $a\in\Sigma$. Any such DFA is said to be a \emph{subautomaton} of $\mathrsfs{A}$. Clearly, if $\mA$ is synchronizing, then so is each of its subautomata.

An equivalence $\pi$ on the state set $Q$ of $\mathrsfs{A}$ is called a \emph{congruence} if $(p,q)\in\pi$ implies $\bigl(\delta(p,a),\delta(q,a)\bigr)\in\pi$ for all $p,q\in Q$ and all $a\in\Sigma$. If $\pi$ is a congruence and $q$ is a state, $[q]_\pi$ stands for the $\pi$-class containing $q$. The \emph{quotient} {$\mathrsfs{A}/\pi$} is the DFA $\langle Q/\pi,\Sigma,\delta_\pi\rangle$, where $Q/\pi:=\{[q]_\pi\mid q\in Q\}$ and the function $\delta_\pi$ is defined by the rule $\delta_\pi([q]_\pi,a):=[\delta(q,a)]_\pi$. Again, if $\mA$ is synchronizing, then so is each of its quotients.

A state $s$ of a DFA is called a \emph{sink} if $s\dt a=s$ for each input letter~$a$. Clearly, if a \san\ has a sink, then the sink must be unique, and every \sw\ must bring the automaton to the sink.

Let $\mathrsfs{A}=\langle Q,\Sigma\rangle$ be a DFA and $p,q\in Q$. We say that $q$ is \emph{reachable} from $p$ if there is a word $w\in\Sigma^*$ such that $p\dt w=q$.
A DFA is called \emph{\scn} if every state in it is reachable from every other state. A \scan\ with more than one state cannot have any sink.

It is actually a part of \sa\ folklore that estimating \rt{} for a ``sufficiently robust'' class of \sa\ reduces to two special cases: the case of automata with a unique sink and the case of \sca. Here it is convenient to employ this folklore fact in the following form, which is a special case of \cite[Proposition 2.1]{Volkov:2009}.

\begin{Lemma}
\label{lem:reduction} Let\/ $\mathbf{C}$ be any class of automata closed under taking subautomata and quotients, and let\/ $\mathbf{C}_n$ stand for the class of all $n$-automata in~$\mathbf{C}$. If each \san\ in\/ $\mathbf{C}_n$ which either is strongly connected or possesses a unique sink has a \sw\ of length $n-1$, then the same holds for all \sa\ in\/ $\mathbf{C}_n$.
\end{Lemma}

\begin{Prop}
\label{prop:binary}
The \rt\ of every \snan\ with two idempotent input letters does not exceed $n-1$.
\end{Prop}

\begin{proof}
The class of automata with two idempotent input letters is closed under taking subautomata and quotients. Hence Lemma~\ref{lem:reduction} applies, and it suffices to verify the bound of Proposition~\ref{prop:binary} for \sa\ in this class which either are strongly connected or have a unique sink.

First, consider the \scn\ case. In this case we will prove a much stronger fact: the \rt\ of any \scn\ \snan\ $\mF=\langle Q,\Sigma\rangle$ with $|\Sigma|=2$ and the letters in $\Sigma$ being idempotents does not exceed $\max\{1,n-1\}$. We may suppose that $n>1$; otherwise, there is nothing to prove. Take a state $q_0\in Q$. There must be a letter $a\in\Sigma$ such that $p_0:=q_0\dt a\ne q_0$ as $q_0$ would be a sink otherwise. Since $\mA$ is \scn, there exists a word $w\in\Sigma^*$ such that $p_0\dt w=q_0$. Let $w$ be chosen to be a word of minimum length with the latter property. Since $p_0\dt a=q_0\dt a^2=q_0\dt a=p_0$, we see that $w\ne aw'$ because otherwise the suffix $w'$ would be a shorter word with the property $p_0\dt w'=q_0$. Thus, denoting by $b$ the other letter in $\Sigma$, we conclude that $w$ starts with $b$ and, moreover, the occurrences of $b$ and $a$ alternate in $w$. The last letter of $w$ cannot be $a$ because $q_0$ would be a fixed point of $a$ otherwise, and this is not the case. Summarizing, we see that $aw$ starts with $a$, ends with $b$, and $a$ and $b$ alternate in $aw$, that is, $aw=(ab)^k$ for some $k\ge1$.

If $k>1$, let $q_i:=q_0\dt(ab)^i$ and $p_i:=q_0\dt(ab)^ia$ for $i=1,\dots,k-1$. Then for each $i=0,\dots,k-1$, we have
\begin{equation}
\label{eq:cycle}
q_i\dt a=p_i,\ \ q_i\dt b=q_i,\ \ p_i\dt a=p_i,\ \ p_i\dt b=q_{i+1(\!\bmod k)}.
\end{equation}
From this, it is easy to deduce that $q_0\dt v\ne q_1\dt v$ for each word $v\in\Sigma^*$. Indeed, write $v$ as a product of powers of $a$ and $b$ and let $\ell$ be the number of `switches' from a power of $a$ to a power of $b$; that is, $\ell$ is the number of occurrences of the word $ab$ as a factor in $v$. Then a direct calculation based on \eqref{eq:cycle} gives the following values of $q_0\dt v$ and $q_1\dt v$: if $v$ ends with $a$, then $q_0\dt v=p_{\ell(\!\bmod k)}$ and $q_1\dt v=p_{\ell+1(\!\bmod k)}$; if $v$ ends with~$b$, then $q_0\dt v=q_{\ell(\!\bmod k)}$ and $q_1\dt v=q_{\ell+1(\!\bmod k)}$. Hence, the automaton $\mF$ is not synchronizing, a contradiction.

Thus, $k=1$, and we have $q_0\dt a=p_0,\ \ q_0\dt b=q_0,\ \ p_0\dt a=p_0,\ \ p_0\dt b=q_0$, that is, the set $\{q_0,p_0\}$ is closed under the action of letters in $\Sigma$. Since the automaton $\mF$ is \scn, this is only possible provided that $n=2$ and $Q=\{q_0,p_0\}$. The automaton $\mathrsfs{F}$ is then nothing but the classical flip-flop, see Fig.~\ref{fig:flip-flop}. Of course, every word of length 1 is a \sw\ for the flip-flop.
\begin{figure}[h]
\begin{center}
  \unitlength=3.5pt
    \begin{picture}(25,12)(0,-5)
    \gasset{curvedepth=4}
    \node(A1)(0,0){$p_0$}
    \drawloop[loopangle=180](A1){$a$}
    \node(A2)(25,0){$q_0$}
    \drawloop[loopangle=0](A2){$b$}
    \drawedge(A1,A2){$b$}
    \drawedge(A2,A1){$a$}
    \end{picture}
\caption{Filp-flop}\label{fig:flip-flop}
\end{center}
\end{figure}

Now consider the case of automata with a unique sink. Let $\mS=\langle Q,\Sigma\rangle$ be a \snan\ with two idempotent input letters and a unique sink, which we denote $z$. If $p$ and $q$ are two different states in $Q$, we say that $p$ is a \emph{predecessor} of $q$ in $\mS$ if $q$ is the image of $p$ under the action of a letter in $\Sigma$. First, we show that there exists a state without predecessors in $\mS$. Arguing by a contradiction, assume that every state has a predecessor in $\mS$. Since the number of states is finite, this assumption implies that for some $k>1$, there is a sequence $S$ of states $q_0,q_1,\dots,q_{k-1}$ such that $q_i$ is a predecessor of $q_{i+1(\!\bmod k)}$ for each $i=0,1,\dots,k-1$. Clearly, $z\notin S$, as the sink is not a predecessor of any state. Now take any $q_i\in S$ and let $a\in\Sigma$ be such that $q_{i-1(\!\bmod k)}\dt a=q_i$. Then $q_i\dt a=q_i$ since the letter $a$ is idempotent, and if $b$ is the other letter in $\Sigma$, then $q_i\dt b=q_{i+1(\!\bmod k)}$. We see that both $q_i\dt a$ and $q_i\dt b$ belong to $S$, that is, the set $S$ is closed under the action of letters in $\Sigma$. Hence, $z$ is not reachable from any state in $S$, while any \sw\ must bring all states in $Q$, including those in $S$, to $z$. This is a contradiction.

We show that $\mS$ has a \sw\ of length $n-1$ by induction on $n$. The induction basis $n=1$ is obvious; now assume that $n>1$. As shown in the preceding paragraph, there is a state $q\in Q$ that has no predecessors in $\mS$. Clearly, $q\ne z$ because $\mS$ is synchronizing and the sink $z$ must be reachable from every state in $Q$, whence $z$ must have a predecessor. Then the subautomaton $\mS'$ obtained by the restriction of the transition function of $\mS$ to the set $Q':=Q\setminus\{q\}$ is synchronizing and has $n-1$ states, two idempotent input letters, and a unique sink. The induction assumption applies to $\mS'$, whence there exists a word $w$ of length $n-2$ such that $Q'\dt w=\{z\}$. Since $q\ne z$, there exists a letter $a\in\Sigma$ such that $q\dt a\ne q$, that is, $q\dt a\in Q'$. Therefore, $Q\dt aw=\{z\}$, whence $aw$ is a \sw\ of length $n-1$ for $\mS$.
\end{proof}

Finally, we exhibit a series of $n$-automata $\mI_n$, $n=3,4,\dotsc$, showing that the bound of Proposition~\ref{prop:binary} is tight. The state set of $\mI_n$ is $\{1,2,\dots,n\}$; the action of the input letters $a$ and $b$ is defined as follows:
\[
i\dt a=\begin{cases}
i &\text{if $i$ is odd or $i=n$},\\
i{+}1 &\text{if $i$ is even and $i<n$};
\end{cases}\quad
i\dt b=\begin{cases}
i &\text{if $i$ is even or $i=n$},\\
i{+}1 &\text{if $i$ is odd and $i<n$}.
\end{cases}
\]
See Fig.~\ref{fig:example} for an illustration. By the construction, the letters $a$ and $b$ act as idempotent selfmaps. Clearly, the state $n$ is a unique sink of $\mI_n$. It is easy to see that $\mI_n$ is synchronizing, and the \rt\ of $\mI_n$ is at least $n-1$ since $n-1$ is the length of the shortest word that brings the state 1 to the sink.
\begin{figure}[th]
\begin{center}
\unitlength=.9mm
\begin{picture}(104,21)(5,5)
\node(A)(10,10){1}
\node(B)(32,10){2}
\node(C)(54,10){3}
\node(D)(76,10){4}
\multiput(85,10)(3,0){3}{\circle*{.75}}
\node(E)(100,10){$n$}
\drawedge(A,B){$b$}
\drawedge(B,C){$a$}
\drawedge(C,D){$b$}
\drawloop[ELpos=50,loopangle=90](A){$a$}
\drawloop[ELpos=50,loopangle=90](B){$b$}
\drawloop[ELpos=50,loopangle=90](C){$a$}
\drawloop[ELpos=50,loopangle=90](D){$b$}
\drawloop[ELpos=50,loopangle=90](E){$a,b$}
\end{picture}
\caption{The automaton $\mI_n$}
\label{fig:example}
\end{center}
\end{figure}

Quite interestingly, Gusev~\cite{Gusev:2013} has constructed for each odd~$n$, a binary \snan\ with \rt\ $\frac{n^2 -3n + 4}{2}$, in which one letter is idempotent and the other acts almost as an idempotent in the sense that it does not fix exactly one state in its image. Fig.~\ref{fig:Gusev} shows the 7-state automaton from Gusev's series; observe that it differs from $\mI_7$ by a single transition!
\begin{figure}[t]
\begin{center}
\unitlength .75mm
\begin{picture}(50,90)(-30,-85)
\gasset{Nw=14,Nh=14,Nmr=7,loopdiam=10}
\node(n0)(0,-5){1}
\node(n1)(-30,-20){7} \node(n2)(30,-20){2}
\node(n3)(-35,-50){6} \node(n4)(35,-50){3}
\node(n5)(-17,-75){5} \node(n6)(17,-75){4}
\drawedge[ELdist=2.0](n1,n0){$b$}
\drawedge[ELdist=1.5](n2,n4){$a$}
\drawedge[ELdist=1.7](n0,n2){$b$}
\drawedge[ELdist=1.7](n3,n1){$a$}
\drawedge[ELdist=1.7](n4,n6){$b$}
\drawedge[ELdist=1.7](n6,n5){$a$}
\drawedge[ELdist=1.7](n5,n3){$b$}
\drawloop[ELdist=1.5,loopangle=-90](n0){$a$}
\drawloop[ELdist=1.5,loopangle=205](n2){$b$}
\drawloop[ELdist=1.5,loopangle=-30](n1){$a$}
\drawloop[ELdist=2.4,loopangle=160](n4){$a$}
\drawloop[ELdist=1.5,loopangle=20](n3){$b$}
\drawloop[ELdist=1.5,loopangle=120](n6){$b$}
\drawloop[ELdist=1.5,loopangle=60](n5){$a$}
\end{picture}
\caption{Gusev's automaton for $n=7$}
\label{fig:Gusev}
\end{center}
\end{figure}
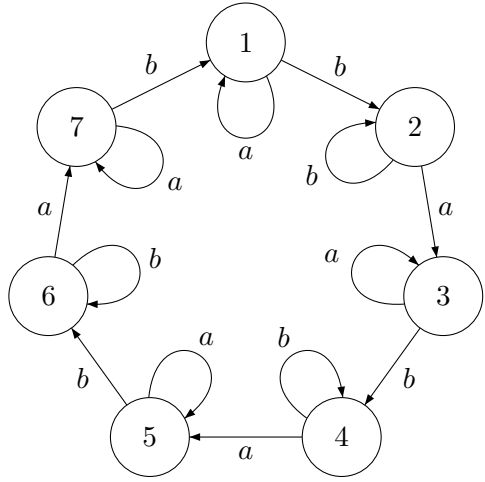
Thus, even a minimum possible step out of idempotency causes a drastic leap in terms of the value of \rt.

\section{Conclusion}

We have described a transformation $\mA\mapsto\mathbb{H}(\mA)$ that converts an arbitrary DFA with $n$ states and $k$ input letters into a DFA with $2n$ states and $k+1$ idempotent input letters of rank $n$. The transformation preserves the property of being synchronizing and doubles the \rt. Observe in passing that the transformation preserves several other properties relevant in the context of synchronization: say, $\mA$ is \scn\ or proper if and only if so is $\mathbb{H}(\mA)$.

Informally, our main result (Theorem~\ref{thm:main}) shows that synchronization problems for the class of DFAs with idempotent input letters are as difficult as in the general case. This informal statement can be precisely expressed in the language of computational complexity---one can use Theorem~\ref{thm:main} to transfer numerous hardness results collected in the theory of \sa\ (see \cite[Section~2]{KV} for an overview) to automata whose letters act as idempotent selfmaps. Theorem~\ref{thm:main} also implies that a quadratic in $n$ upper bound on the \rt\ of \sna\ with idempotent input letters would lead to a quadratic upper bound in the general case---here we recall that the best upper bound known so far is cubic. In this connection, we mention a result by Rystsov~\cite{Rystsov:2000} who established the upper bound $2(n-1)^2$ on the \rt\ of \sna\ whose input letters are either permutations or idempotents of rank $n-1$.

\paragraph{Acknowledgements.} The author is very much indebted to the anonymous referees of the previous version for their valuable remarks (in particular, for spotting inaccuracies in the original formulations of Lemma~\ref{lem:idempotent} and Corollary~\ref{cor:cerny} and suggesting an improvement in Theorem~\ref{thm:main}) and for providing several relevant references.

\bibliographystyle{plain}
\bibliography{abbrevs,idempotent_rev}

\end{document}